\documentclass[aps,pra,english,twocolumn,notitlepage,superscriptaddress,showpacs,nofootinbib]{revtex4-1}

\usepackage{lipsum,mwe}
\usepackage{amsmath,amsfonts,amssymb,caption,color,epsfig,graphics,graphicx,hyperref,latexsym,mathrsfs,revsymb,theorem,url,verbatim,enumerate,epstopdf,tikz,float,multirow,booktabs,appendix}
\hypersetup{colorlinks,linkcolor={blue},citecolor={red},urlcolor={blue}}
\usetikzlibrary{arrows, decorations.markings}

\usepackage{epstopdf}
\tikzstyle{vecArrow} = [thick, decoration={markings,mark=at position
	1 with {\arrow[semithick]{open triangle 60}}},
double distance=1.4pt, shorten >= 5.5pt,
preaction = {decorate},
postaction = {draw,line width=1.4pt, white,shorten >= 4.5pt}]
\tikzstyle{innerWhite} = [semithick, white,line width=1.4pt, shorten >= 4.5pt]

\newtheorem{definition}{Definition} 
\newtheorem{proposition}{Proposition} 
\newtheorem{lemma}{Lemma}

\newtheorem{theorem}{Theorem} 
\newtheorem{corollary}[definition]{Corollary}
\newtheorem{conjecture}[definition]{Conjecture}

\newtheorem{remark}[definition]{Remark}
\newtheorem{example}{Example} 
\newtheorem{question}[definition]{Question}

\def\bcj{\begin{conjecture}}
	\def\ecj{\end{conjecture}}
\def\bcr{\begin{corollary}}
	\def\ecr{\end{corollary}}
\def\bd{\begin{definition}}
	\def\ed{\end{definition}}
\def\bea{\begin{eqnarray}}
	\def\eea{\end{eqnarray}}
\def\bem{\begin{enumerate}}
	\def\eem{\end{enumerate}}
\def\bex{\begin{example}}
	\def\eex{\end{example}}
\def\bim{\begin{itemize}}
	\def\eim{\end{itemize}}
\def\bl{\begin{lemma}}
	\def\el{\end{lemma}}
\def\bma{\begin{bmatrix}}
	\def\ema{\end{bmatrix}}
\def\bpf{\begin{proof}}
	\def\epf{\end{proof}}
\def\bpp{\begin{proposition}}
	\def\epp{\end{proposition}}
\def\bqu{\begin{question}}
	\def\equ{\end{question}}
\def\br{\begin{remark}}
	\def\er{\end{remark}}
\def\bt{\begin{theorem}}
	\def\et{\end{theorem}}

\def\squareforqed{\hbox{\rlap{$\sqcap$}$\sqcup$}}
\def\qed{\ifmmode\squareforqed\else{\unskip\nobreak\hfil
		\penalty50\hskip1em\null\nobreak\hfil\squareforqed
		\parfillskip=0pt\finalhyphendemerits=0\endgraf}\fi}
\def\endenv{\ifmmode\;\else{\unskip\nobreak\hfil
		\penalty50\hskip1em\null\nobreak\hfil\;
		\parfillskip=0pt\finalhyphendemerits=0\endgraf}\fi}
\newenvironment{proof}{\noindent \textbf{{Proof.~} }}{\qed}
\def\Dbar{\leavevmode\lower.6ex\hbox to 0pt
	{\hskip-.23ex\accent"16\hss}D}
\makeatletter
\def\url@leostyle{%
	\@ifundefined{selectfont}{\def\UrlFont{\sf}}{\def\UrlFont{\small\ttfamily}}}
\makeatother
\def\bcj{\begin{conjecture}}
	\def\ecj{\end{conjecture}}
\def\bcr{\begin{corollary}}
	\def\ecr{\end{corollary}}
\def\bd{\begin{definition}}
	\def\ed{\end{definition}}
\def\bea{\begin{eqnarray}}
	\def\eea{\end{eqnarray}}
\def\bem{\begin{enumerate}}
	\def\eem{\end{enumerate}}
\def\bex{\begin{example}}
	\def\eex{\end{example}}
\def\bim{\begin{itemize}}
	\def\eim{\end{itemize}}
\def\bl{\begin{lemma}}
	\def\el{\end{lemma}}
\def\bpf{\begin{proof}}
	\def\epf{\end{proof}}
\def\bpp{\begin{proposition}}
	\def\epp{\end{proposition}}
\def\bqu{\begin{question}}
	\def\equ{\end{question}}
\def\br{\begin{remark}}
	\def\er{\end{remark}}
\def\bt{\begin{theorem}}
	\def\et{\end{theorem}}

\def\btb{\begin{tabular}}
	\def\etb{\end{tabular}}

\begin{document}
 
	\title{ Local discrimination of generalized Bell states via commutativity}

	\author{Mao-Sheng Li}
	\affiliation{ School of Mathematics,
		South China University of Technology, Guangzhou
		510641,  China}
	
	\author{Fei Shi}
	\affiliation{School of Cyber Security,
		University of Science and Technology of China, Hefei, 230026, People's Republic of China}

	\author{Yan-Ling Wang}
	\email{wangylmath@yahoo.com}
	\affiliation{ School of Computer Science and Technology, Dongguan University of Technology, Dongguan, 523808, China}

	\begin{abstract}
		We study the distinguishability of   generalized Bell states under   local operations and classical communication. We introduce    the concept of maximally commutative set (MCS), subset of generalized Pauli matrices whose elements   are mutually commutative, and there is no other generalized Pauli matrix that is  commute with all the elements of this set. We find that  MCS can be considered as a detector  for local distinguishability of set $\mathcal{S}$  of    generalized Bell states.  In fact, we get an efficient criterion. That is, if  the difference set of $\mathcal{S}$ is disjoint with or completely contain in  some MCS, then the set  $\mathcal{S}$ is locally distinguishable.  Furthermore, we give a    useful characterization of   MCS for arbitrary dimension, which provides great convenience for   detecting the local discrimination of generalized Bell states. Our method can be generalized to more general settings which contains lattice qudit basis.  Results in [\href{https://journals.aps.org/prl/abstract/10.1103/PhysRevLett.92.177905}{Phys. Rev. Lett. \textbf{92}, 177905 (2004)}], [\href{https://journals.aps.org/pra/abstract/10.1103/PhysRevA.92.042320}{Phys. Rev. A \textbf{92}, 042320 (2015)}] and a recent work [\href{https://arxiv.org/abs/2109.07390}{arXiv: 2109.07390}] can   be deduced as special cases of our
		result.

		\begin{description}
			\item[PACS numbers] 03.67.Hk,03.65.Ud
		\end{description}
	\end{abstract}                            
	\maketitle
	
	\section{Introduction}
	Quantum states discrimination is a fundamental task in quantum information processing. It is well known that a set of quantum states can be perfectly distinguished by global measurement if and only if the states of given set    are mutually orthogonal \cite{nils04}.  However,   our quantum states are usually distributed in  composite systems  with long distances, so only local operations and classical communication (LOCC) are allowed.  In such setting, a state is chosen from a known orthogonal set  of quantum states in a  composite systems and  the  task is to identify the state  under LOCC.   If the task can be  accomplished perfectly, we say that the set is \emph{locally distinguishable},
	otherwise, \emph{locally indistinguishable}. If an orthogonal set is locally indistinguishable, we also called that the set presents some kind of nonlocality \cite{Bennett99} in the sense that more  quantum information  could be inferred from global measurement than that from local operations.  Any two orthogonal multipartite states are showed to be locally distinguishable \cite{Walgate-2000}. 	Bennett \emph{et al.} \cite{Bennett99} presented the first example 
	of  orthogonal product states that are locally indistinguishable which reveals
the phenomenon of “quantum nonlocality without entanglement”.
	Results on the local distinguishability of quantum states have  been practically applied in quantum cryptography primitives such as   data hiding \cite{Terhal01,DiVincenzo02} and secret sharing \cite{Markham08,Rahaman15,WangJ17}.
	
	For general orthogonal sets of quantum states, it is   difficult to give a complete characterization of whether they are locally distinguishable or not. Therefore, most studies  (See \cite{Bennett99,ben99u,Walgate-2000, Ghosh-2001, walgate-2002,divin03,HSSH,Ghosh-2004,rin04,Watrous-2005,fan-2005,Nathanson-2005, Hayashi-etal-2006,nis06,Duan2007,Duan-2009,feng09,Bandyo-2011,BGK-2011,Yu-Duan-2012,Nathanson13,BN-2013,Cosentino-2013, Yus15,  B-IQC-2015,childs13,Cosentino-Russo-2014, Ha21,B-IQC-2015,  Li15, Yang13,zhang14,wang15,zhang15,zhang16-1,Xu-16-2,zhang16,Wang-2017-Qinfoprocess,Zhang-Oh-2017,halder, Xu20b, Li20,  Li18,Halder1909,  Halder20c,Halder19,Zhang1906,Shi20S,Tian20,Wang21, Banik21,Ha21} for an incomplete list)   focus on two extreme cases: sets of product states or sets of maximally entangled states.  In this paper, we restrict ourselves to the settings of maximally entangled cases.
	
Bell states are the most famous maximally entangled states and their local distinguishability has been well understood.  In fact, any two Bell states are locally distinguishable but any three or four are not \cite{Ghosh-2001}.  Nathanson \cite{Nathanson-2005} showed that any three maximally entangled states in $\mathbb{C}^3\otimes\mathbb{C}^3$ can be locally distinguished. Moreover, any $l>d$   maximally   entangled states in $\mathbb{C}^d\otimes \mathbb{C}^d$  are known to be locally indistinguishable \cite{Nathanson-2005}. Therefore, it is interesting to consider whether set of  maximally entangled states  with cardinality $l\leq d$ can be locally distinguishable or not?  Interestingly,   using the fact that  applying local unitary operation does not change the local distinguishability,   Fan \cite{fan-2005} showed that any $l$ generalized Bell states (GBSs) in $\mathbb{C}^d\otimes \mathbb{C}^d$ are locally distinguishable if $(l-1)l\leq 2d$ provided that  $d$ is a prime number.     Fan's  result was extended by  Tian \textit{et al.} to the prime power dimensional quantum system in \cite{Tian15} where they restricted themselves to  the mutually commuting qudit lattice states. Since Fan's result,  there has been lots of  works \cite{Tian15_2,Tian15,Tian16,Singal15,Wang17,Yuan20,Yang21,Hashimoto21} paid attention to the locally distinguishability of GBSs. However, the complete classification of  local distinguishablity of GBSs is still difficult to achieve.    On the other hand, set of  GBSs is an important and special subset of maximally entangled states,  which makes  the problem of local distinguishability of GBSs   important and interesting.  Motivated by a recent work \cite{Yuan21},  we find that the local distinguishability of GBSs can be detected by maximally commutative set (MCS) of GBSs.
	
	The rest of this article is organized as follows.  In Sec. \ref{sec:Review},  we introduce the matrices representation of generalized Bell states. Then we give a brief review of some known results on the sufficient conditions of locally distinguishable set of GBSs.  In  Sec. \ref{Sec:MainResult}, we present the definition of maximally commutative set    and show that it is useful for judging the locally distinguishability. After that, we present some examples of MCS and study the properties of general MCS.   Finally, we draw a conclusion and presented some questions in   Sec. \ref{Sec:Con_Dis}.

	\noindent	
	
	\section{A Review of  local distinguishability of $\mathrm{GBSs}$}\label{sec:Review}
	Throughout this paper, we will use the following notations.	Let $d\geq 2$ be an integer. Denote    $\mathbb{Z}_d$ to be the ring defined over $\{0,1,\cdots,d-1\}$ with the sum operation  ``$+$" (here $i + j$ should be equal to the element $(i+j)\  \mathrm{ mod } \ d$) and  multiplication operation (computing the usual multiplication first then taking module $d$).
	Consider a bipartite quantum system $\mathcal{H}_A\otimes \mathcal{H}_B$ with both local dimensions equal to  $d$.  Suppose that $\{|0\rangle, |1\rangle, \cdots, |d-1\rangle \}$  is the    computational basis of a single qudit.  Under this computational basis, the standard maximally entangled state in this system can be expressed as  ${|\Psi_{00}\rangle=\frac{1}{\sqrt{d}} \sum_{i=0}^{d-1}|ii\rangle}.$  Generally, any maximally entangled state can be written in the form  $|\Psi_U\rangle=(I\otimes U)|\Psi_{00}\rangle$   for   some unitary matrix  $ U $ of dimensional $d$.  We often call $U$ the defining unitary matrix  of the maximally entangled state $|\Psi_U\rangle$. To define the generalized Bell states,  we define the following two operations
	$$X_d=\displaystyle\sum_{i=0}^{d-1}|i+1 \text{ mod } d\rangle\langle i|, \text{ and } Z_d=\displaystyle\sum_{i=0}^{d-1}\omega^i|i\rangle\langle i|,$$
	where $\omega=e^{\frac{2\pi \sqrt{-1} }{d}}$. Then the following $d^2$  orthogonal  MESs are  called as    generalized Bell states:
	\begin{equation}\label{eq1}
		\{|\Psi_{m,n}\rangle=(I\otimes X_d^mZ_d^n)|\Psi_{00}\rangle\big|  m,n\in \mathbb{Z}_d\}.
	\end{equation}
	And matrices in 	$\{  X_d^mZ_d^n \big  |\  m,n \in \mathbb{Z}_d\}$  are called the generalized Pauli matrices (GPMs). 
	For simplicity, we also  use $X$ and $Z$ to represent $X_d$ and $Z_d$  when the dimension is known.
	Due to  the one-one correspondence of MES and its defining unitary matrix, for convenience, we will treat the following three sets equally without distinction
	\begin{equation*}
		\mathcal{S}:= \{|\Psi_{m_i,n_i}\rangle  \}_{i=1}^l=\{X^{m_i}Z^{n_i}  \}_{i=1}^l=\{(m_i,n_i) \}_{i=1}^l.
	\end{equation*}		
	Our aim  in this paper is to provide some sufficient  condition such that the set $\mathcal{S}$ is locally distinguishable.  Now we gave a brief review of the relative results.

	Fan \cite{fan-2005}  noted that if all $m_i$ ($i=1,\cdots,l$) are
	distinct, the set $\mathcal{S}$ can be  locally  distinguished and set with this property is called $F$-type \cite{Hashimoto21}. For each $\alpha\in \mathbb{Z}_d$, defining $H_\alpha$ be the matrix whose $jk$ entry is $w^{-jk-\alpha s_k}/\sqrt{d}$ for $j,k=1,\cdots d-1$  and $s_k:=\sum_{i=k}^{d-1} i$. Then $H_\alpha$ is a unitary matrix and   $H_\alpha \otimes  H_\alpha^t$ transfers $|\Psi_{m_i,n_i}\rangle $ to  $|\Psi_{\alpha m_i+n_i,-m_i}\rangle. $ He found that if $d$ is a prime and $(l-1)l\leq 2d$, there exists an $\alpha$ such that  $H_\alpha \otimes  H_\alpha^t$ can transfer $\mathcal{S}$ to a  set of $F$-type.
	
There is a useful sufficient condition for local distinguishability  of   general maximally entangled states. Denote  $\mathcal{S}$ as the defining unitary matrices set, if there exists some nontrivial vector $|v\rangle \in \mathbb{C}^d$ such that
	\begin{equation}\label{eq:vector}
		\langle v|   U^\dagger V|v\rangle =0
		\end{equation} for 
	any different $U,V\in \mathcal{S}$, then the set of maximally entangled states corresponding to $\mathcal{S}$ is  one-way   distinguishable (hence locally distinguishable) \cite{Ghosh-2004,Nathanson13}.  If the set $\mathcal{S}$ is $F$-type, the vector $|v\rangle $ can be chosen as any vector of the computational basis, i.e., $|i\rangle, i\in \mathbb{Z}_d$. We define difference set $\Delta \mathcal{S}$ of  $\mathcal{S}:=\{U_i|i=1,2,\cdots,l\}$  as 
	\begin{equation}\label{eq:difference}
		\Delta \mathcal{S}=\{ U_i^\dagger U_j \ |\  1\leq i<j\leq l \}.
	\end{equation}
	Noticing that $$(X^{m_i}Z^{n_i})^\dagger X^{m_j}Z^{n_j}= \omega^{-(m_j-m_i)n_i} X^{m_j-m_i}Z^{n_j-n_i}.$$ Up to a phase, we can identify $\Delta \mathcal{S}$ as the set 
	$\{(m_j-m_i,n_j-n_i)| 1\leq i< j \leq l\}.$
	In order to find some nonzero vector $|v\rangle$ such that Eqs. \eqref{eq:vector} are satitsfied, the following lemma is important (See also in Ref. \cite{Yuan21}). 	
	
		\begin{lemma}\label{lemma:eig} 
		For two unitary matrices $U$ and $V$, if  they satisfy 	$UV=zVU$ where $z$ is a complex number and  are not commutative, i.e., $z\neq 1$,
		then each eigenvector $|v\rangle$ of $V$ satisfies $\langle v|U|v\rangle=0$. 		
	\end{lemma}
In fact, suppose that  $V|v\rangle=\lambda |v\rangle$ where $\overline{\lambda}\lambda  =1$.  We also have $\langle v| V^\dagger =\overline{\lambda} \langle v|.$ Therefore, $\langle v|U|v\rangle=\langle v|\overline{\lambda} U\lambda|v\rangle= \langle u|V^\dagger U V|u\rangle= z\langle v|U|v\rangle.$ Hence, $\langle v|U|v\rangle=0$ as $z\neq 1.$   A pair of unitaries that satisfy the first condition are called  Weyl commutative.

Fortunately, any pair of generalized Pauli matrices are Weyl commutative. In fact, for two pairs of $(m_i,n_i)$ and $(m_j,n_j)$ in $\mathbb{Z}_d\times  \mathbb{Z}_d$, we always have 
$$X^{m_i}Z^{n_i}X^{m_j}Z^{n_j}=\omega^{m_jn_i-m_in_j}X^{m_j}Z^{n_j}X^{m_i}Z^{n_i}.$$
Moreover, $X^{m_i}Z^{n_i}$ and $X^{m_j}Z^{n_j}$ are     commutative if and only if ${m_jn_i-m_in_j}\equiv 0\  \mathrm{ mod } \ d$. This condition can be formulated as  the detminant equation 
		\begin{equation}\label{eq:detdef}		       \left|\begin{array}{cc}
		m_i&n_i\\
		m_j&n_j
	\end{array}\right|\equiv 0 \mod d.
\end{equation}
We also call  that $(m_i,n_i)$ and $(m_j,n_j)$ are commutative if this condition is satisfied.
	
	For any GBS set $\mathcal{S}$,
	if there is a generalized Pauli matrix $V$ which is not commutative to every GPM $U\in\Delta \mathcal{S}$, by Eqs. \eqref{eq:vector},  \eqref{eq:difference}, and Lemma \ref{lemma:eig},  
	  each eigenvector $|v\rangle$ of $V$ satisfies $\langle v|U|v\rangle=0$
	and therefore the set $\mathcal{S}$  is locally distinguishable.

	Let $m, n\in\mathbb{Z}_{d}$, denote the solution set of the  following congruence equation   by $S(m,n)$,
\begin{eqnarray}\label{cogru2.1}
	nx-my=0 \mod d.
\end{eqnarray}
Therefore,  $S(m,n)$ denote the set of elements  in $ \mathbb{Z}_d\times \mathbb{Z}_d$ that are commute with $(m,n).$	In order to check whether is there any GPM $V$ that do not commute with all the elements in $\Delta \mathcal{S}.$ The authors of Ref. \cite{Yuan21}  defined the set
	\begin{eqnarray*}
		\mathcal{D}(\mathcal{S})\triangleq (\mathbb{Z}_d\times \mathbb{Z}_d) \setminus \bigcup_{(m,n)\in\Delta\mathcal{S}} S(m,n).
	\end{eqnarray*}
 By definition, $\mathcal{D}(\mathcal{S})$ denotes the set of all elements in $\mathbb{Z}_d\times \mathbb{Z}_d$ that do not commute with all the elements in $\Delta \mathcal{S}.$ 
Under this definition, they proved the following results.
	
	\begin{theorem}[See Ref. \cite{Yuan21}]\label{th3.1}
		Let $\mathcal{S}=\{(m_{i},n_{i})|4\le i\le l\le d\}$ be a GBS set in  $\mathbb{C}^{d}\otimes\mathbb{C}^{d}$,
		then the set $\mathcal{S}$ is local distinguishable when any of the following conditions is true.
		\begin{enumerate}
			\item[{\rm(1)}] The discriminant set $\mathcal{D}(\mathcal{S})$ is not empty.
			\item[{\rm(2)}] The set $\Delta\mathcal{S}$ is commutative.
			\item[{\rm(3)}] The dimension $d$ is a composite number, and for each $(m,n)\in\Delta\mathcal{S}$, $m$ or $n$ is invertible in $\mathbb{Z}_d$.
		\end{enumerate}
	\end{theorem}
	
		The nonemptyness of $\mathcal{D}(\mathcal{S})$ implies the local distinguishability of $\mathcal{S}$.  Therefore, the set $\mathcal{D}(\mathcal{S})$ can be called a $\it{discriminant\ set}$ of $\mathcal{S}$.

	\section{Detector for local distinguishability  of $\mathrm{GBSs}$}\label{Sec:MainResult}	 	
	In the first case of  Theorem \ref{th3.1}, the nonempty of the discriminant set $\mathcal{D}(\mathcal{S})$ is equivalent to that there exists some $(s,t)\in\mathbb{Z}_d\times\mathbb{Z}_d $ such that
	$$ms-nt\neq 0, \ \ \forall (m,n)\in \Delta\mathcal{S}.$$
	That is, $X^sZ^t$ is not commute with  $X^mZ^n$. Therefore, any nonzero eigenvector $|v\rangle$ of  $X^sZ^t$ satisfies 
	$$ \langle v|X^mZ^n|v\rangle=0$$
	from which  one can conlude that the set $\mathcal{S}$ is locally distinguishable. From this point, we can call $X^sZ^t$ as a detector of   local discrimination of GBS. Simply, the ability of the detector $X^sZ^t$ can be defined as the set 
	\begin{equation}\label{PointDe}
		\mathcal{D}e(X^sZ^t)\triangleq (\mathbb{Z}_d\times \mathbb{Z}_d) \setminus   S(s,t).
	\end{equation}
	This denotes the set of all elements in $\mathbb{Z}_d\times \mathbb{Z}_d$ that do not commute with $(s,t)$. Then  the  one-way local distinguishability of  $\mathcal{S}$ can be detected by $X^sZ^t$  if and only if
	$\Delta\mathcal{S}\subseteq \mathcal{D}e(X^sZ^t).$

	In fact, we can introduce a stronger detector by the following observation.  If a set of detectors $\{X^{s_i}Z^{t_i}\}_{i=1}^n$ are commutative, they can share  a common eigenbasis $\{|v_j\rangle\}_{j=1}^d$.  Therefore, if 
	$$\Delta\mathcal{S}\subseteq \bigcup_{i=1}^n\mathcal{D}e(X^{s_i}Z^{t_i})$$ 
	we can also conclude that 
	the set $\mathcal{S}$ is one-way distinguishable. Therefore,  
the more elements of the detected set, the stronger its distinguishingability which motivates the following definition.


	
	
	A subset  $\{X^{s_i}Z^{t_i}\}_{i=1}^n$ 	 of GBSs is call \emph{maximally commutative} if the elements of the given subset are mutually commute and there is no other GBS  which can commute with all the elements of the set. This can be written as the coordinates $\{(s_i,t_i)\}_{i=1}^n\subseteq \mathbb{Z}_d\times \mathbb{Z}_d$ such that $s_it_j=t_is_j$ for every $i,j$ but there is no $(s,t)\in \mathbb{Z}_d\times \mathbb{Z}_d $ such that $s_it=t_is$ for every $i$.
	
For any maximally commutative set of GBSs $\mathcal{C}:=\{X^{s_i}Z^{t_i}\}_{i=1}^n$,
	we defined a detector  as 
	$$ \mathcal{D}e(\mathcal{C}):= \bigcup_{(s,t)\in\mathcal{C}}\mathcal{D}e(X^{s}Z^{t}).$$ 
	Therefore, one conclude that if $\Delta\mathcal{S}\subseteq \mathcal{D}e(\mathcal{C})$, the set $\mathcal{S}$ is one-way distinguishable. On the other hand, one finds that  $\mathcal{D}e(\mathcal{C})$ is equal to 
	$\mathcal{P}_d\setminus \mathcal{C}$
	where $\mathcal{P}_d:=\{X^mZ^n| m,n\in \mathbb{Z}_d\}.$ In fact, every element in $\mathcal{P}_d$ but outside $\mathcal{C}$ must be not commute with one of element $X^sZ^t$ in $\mathcal{C}$. That is, it belongs to $\mathcal{D}e(X^{s}Z^{t})$.  It means that $\mathcal{P}_d\setminus \mathcal{C}\subseteq \mathcal{D}e(\mathcal{C})$. Obviously,  $\mathcal{D}e(\mathcal{C})\subseteq \mathcal{P}_d\setminus \mathcal{C}$. Thus, $\mathcal{D}e(\mathcal{C})=\mathcal{P}_d\setminus \mathcal{C}.$   Therefore, $\Delta\mathcal{S}\subseteq \mathcal{D}e(\mathcal{C})$ if and only if 
	$\Delta\mathcal{S}\cap \mathcal{C}=\emptyset.$ 
Moreover, if $\Delta\mathcal{S}\subseteq  \mathcal{C},$ that is, the elements in  $\Delta\mathcal{S}$ are mutually commutative. By Theorem \ref{th3.1}, the set $\mathcal{S}$ is also locally distinguishable.

	\begin{theorem}\label{thm:maxcom}
		Let $\mathcal{S}$ be a GBS set in  $\mathbb{C}^{d}\otimes\mathbb{C}^{d}$ and   $\mathcal{C}$ be a set of maximally commutative GBS of dimensional $d$.   If  	$\Delta\mathcal{S}\cap \mathcal{C}=\emptyset$ or  $\Delta\mathcal{S}\subseteq  \mathcal{C},$  then the set $\mathcal{S}$ is locally distinguishable (see Fig. \ref{fig:relation} for an intuitive view of the conditions). 
	\end{theorem}

	\begin{figure}[h]
		\centering
		\includegraphics[scale=0.7]{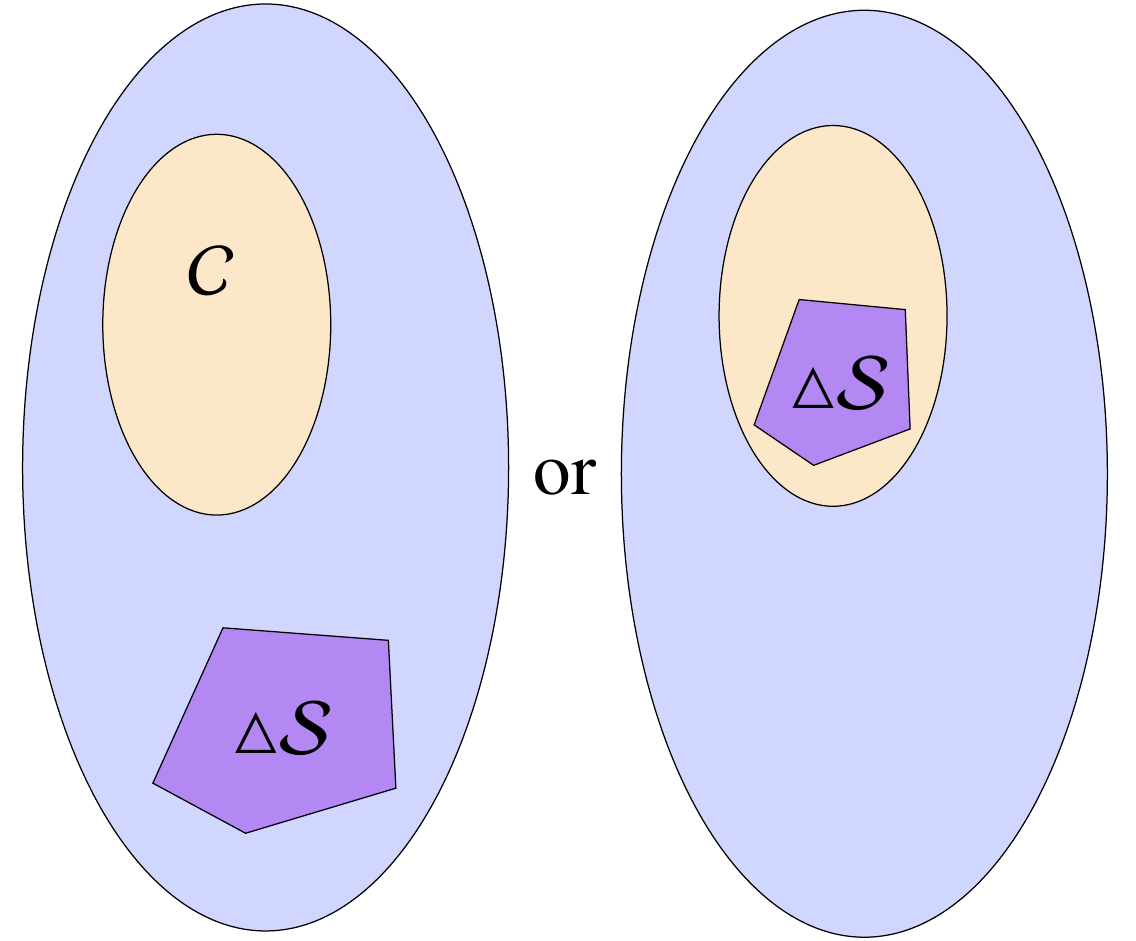}
		\caption{Here $\mathcal{C}$ represents a maximally commutative set of GBS and $\Delta \mathcal{S}$ is the difference set of $\mathcal{S}$. If  $\Delta \mathcal{S}$ and the detector $\mathcal{C}$ are in one of the above relations , then  the set $\mathcal{S}$ is locally distinguishable.  }\label{fig:relation}
	\end{figure}
	The result of Theorem \ref{thm:maxcom} implies  that of Theorem \ref{th3.1}.
	If $\mathcal{D}(\mathcal{S})$  is nonempty (that is the distinguishability of $\mathcal{S}$ can be detected by the first condition of Theorem \ref{th3.1}), then there must exist some maximally commutative  set  $\mathcal{C}$ of GBS satisfies the condition of Theorem \ref{thm:maxcom}. In fact,  the nonempty of the discriminant set $\mathcal{D}(\mathcal{S})$ is equivalent to  the existence of $X^sZ^t$ that do not commute with every element of $\Delta\mathcal{S}$ but such an $X^sZ^t$ can be extended to be a maximally commutative set $\mathcal{C}$ of GBS. As the elements in $\mathcal{C}$ are all commute with $X^sZ^t$, therefore, $\Delta\mathcal{S}\cap \mathcal{C}=\emptyset.$	Moreover,  if $d$ is a composite number, and suppose that $d=pq$ where $p,q\geq 2$ are two integers. Clearly,  $X^p$ is commute with $Z^q$, therefore, they can extend to a maximally commutative set of GBS, said, $\mathcal{C}$. if   $s$ or $t$ is invertible in $\mathbb{Z}_d$, we claim that $X^sZ^t\notin \mathcal{C}.$ In fact,  as $ZX=\omega XZ$, if $s$ is invertable, then $Z^q(X^sZ^t)=\omega^{sq}(X^sZ^t)Z^q\neq (X^sZ^t)Z^q $. If  $t$ is invertible, then $(X^sZ^t)X^p=\omega^{tp}X^p(X^sZ^t)\neq X^p(X^sZ^t).$  It also means that $\Delta \mathcal{S} \cap \mathcal{C}=\emptyset.$ Therefore, if $\Delta\mathcal{S}$ contains those elements one of whose coordinates is invertible in $\mathbb{Z}_d$, then the set $\mathcal{C} $ can detect the one-way distinguishability of  $\mathcal{S}$.
	
	Therefore, it is important to find out all the   maximally commutative sets of GBSs. Now we present some examples in the low dimensional cases.

	\begin{example}\label{ex:d=3}
		There are exactly four classes of maximally commute sets  of GBSs in $\mathbb{C}^3\otimes\mathbb{C}^3.$
		$$
		\begin{array}{rlccc}
			\mathcal{C}_1&=&\{(0,0),&(0,1),&(0,2)\},\\
			\mathcal{C}_2&=&\{(0,0),&(1,0),&(2,0)\},\\
			\mathcal{C}_3&=&\{(0,0),&(1,1),&(2,2)\},\\
			\mathcal{C}_4&=&\{(0,0),&(1,2),&(2,1)\}.
		\end{array}
		$$
	\end{example}
	\begin{example}\label{ex:d=4}
		There are exactly seven classes of maximally commute sets  of GBSs in $\mathbb{C}^4\otimes\mathbb{C}^4.$
		$$
		\begin{array}{rlcccc}
			\mathcal{C}_1&=&\{(0,0),&(0,1),&(0,2),&(0,3)\},\\
			\mathcal{C}_2&=&\{(0,0),&(0,2),&(2,0),&(2,2)\},\\
			\mathcal{C}_3&=&\{(0,0),&(0,2),&(2,1),&(2,3)\},\\
			\mathcal{C}_4&=&\{(0,0),&(1,0),&(2,0),&(3,0)\},\\
			\mathcal{C}_5&=&\{(0,0),&(1,1),&(2,2),&(3,3)\},\\
			\mathcal{C}_6&=&\{(0,0),&(1,2),&(2,0),&(3,2)\},\\
			\mathcal{C}_7&=&\{(0,0),&(1,3),&(2,2),&(3,1)\}.\\
		\end{array}
		$$
	\end{example}

	\begin{example}\label{ex:d=8}
		There are exactly 15 classes of maximally commute sets  of GBSs in $\mathbb{C}^8\otimes\mathbb{C}^8.$ Here we do not write out the coordinate $(0,0)$ which belongs to all the 15 sets.
		{\small		 $$\begin{array}{rlcccccccc}
				\mathcal{C}_1&=&\{	 (0,1),&     (0,2),&     (0,3),&     (0,4),&     (0,5),&     (0,6),&     (0,7)\},\\ 
				\mathcal{C}_2&=&\{	 (0,2),&     (0,4),&     (0,6),&     (4,0),&     (4,2),&     (4,4),&     (4,6)\},\\
				\mathcal{C}_3&=&\{	 (0,2),&     (0,4),&     (0,6),&     (4,1),&     (4,3),&     (4,5),&     (4,7)\},\\
				\mathcal{C}_4&=&\{	 (0,4),&     (2,0),&     (2,4),&     (4,0),&     (4,4),&     (6,0),&     (6,4)\},\\
				\mathcal{C}_5&=&\{	 (0,4),&     (2,1),&     (2,5),&     (4,2),&     (4,6),&     (6,3),&     (6,7)\},\\
				\mathcal{C}_6&=&\{	 (0,4),&     (2,2),&     (2,6),&     (4,0),&     (4,4),&     (6,2),&     (6,6)\},\\
				\mathcal{C}_7&=&\{	 (0,4),&     (2,3),&     (2,7),&     (4,2),&     (4,6),&     (6,1),&     (6,5)\},\\
				\mathcal{C}_8&=&\{	 (1,0),&     (2,0),&     (3,0),&     (4,0),&     (5,0),&     (6,0),&     (7,0)\},\\
				\mathcal{C}_9&=&\{	 (1,1),&     (2,2),&     (3,3),&     (4,4),&     (5,5),&     (6,6),&     (7,7)\},\\
				\mathcal{C}_{10}&=&\{	 (1,2),&     (2,4),&     (3,6),&     (4,0),&     (5,2),&     (6,4),&     (7,6)\},\\
				\mathcal{C}_{11}&=&\{	 (1,3),&     (2,6),&     (3,1),&     (4,4),&     (5,7),&     (6,2),&     (7,5)\},\\
				\mathcal{C}_{12}&=&\{	 (1,4),&     (2,0),&     (3,4),&     (4,0),&     (5,4),&     (6,0),&     (7,4)\},\\
				\mathcal{C}_{13}&=&\{	 (1,5),&     (2,2),&     (3,7),&     (4,4),&     (5,1),&     (6,6),&     (7,3)\},\\
				\mathcal{C}_{14}&=&\{	 (1,6),&     (2,4),&     (3,2),&     (4,0),&     (5,6),&     (6,4),&     (7,2)\},\\
				\mathcal{C}_{15}&=&\{	 (1,7),&     (2,6),&     (3,5),&     (4,4),&     (5,3),&     (6,2),&     (7,1)\}.\\
			\end{array}
			$$
		}
		
	\end{example}

	Using these MCSs, we can show that  Theorem \ref{thm:maxcom} is strictly   powerful than Theorem \ref{th3.1} when the dimension $d=8$. Set $\mathcal{S}:=\{ \mathbb{I},X^5Z^6,X^6Z^3,X^6Z^5,X^7Z^6\}$ whose difference set $\Delta \mathcal{S}$ is 
	{	$$
		\begin{array}{l}
			(5,6),(6,3),(6,5),(7,6),(1,5),\\
			(1,7),(2,0),(0,2),(1,3),(1,1).
		\end{array}$$
	}
	One can check that $\mathcal{D}(\mathcal{S})=\emptyset$, $\Delta \mathcal{S}$ is non-commutative and  neither coordinates of $(2,0)$ are  invertible in $\mathbb{Z}_8$. Therefore, Theorem \ref{th3.1} fails to detect the distinguishability of this set.
	However, one finds that $\Delta \mathcal{S}\cap \mathcal{C}_6=\emptyset.$ That is, the local distinguishablity of $\mathcal{S}$ can be detected by $\mathcal{C}_6$. Moreover, one can check that neither $\Delta \mathcal{S}\cap \mathcal{C}_i=\emptyset$ nor $\Delta \mathcal{S}\subseteq \mathcal{C}_i$ when $i\neq 6$. That is, among the 15 classes, $\mathcal{C}_6$ is the only detector that can detect the local distinguishability of $\mathcal{S}.$  More numerical results comparing the power of   Theorem \ref{th3.1} and  Theorem  \ref{thm:maxcom} can be seen in the figure \ref{fig:Sup} (we randomly generated $N$ sets of  $d$ dimensional GBSs  with cardinality $n$ and find out the numbers $N_{1}$ and $N_{2}$ of  sets whose local distinguishability can be detected by Theorem \ref{th3.1} and  Theorem  \ref{thm:maxcom} respectively. The corresponding successful rates  are defined by $N_1/N$  and $N_2/N$).

	\begin{figure}[h]
		\centering
		\includegraphics[scale=0.59]{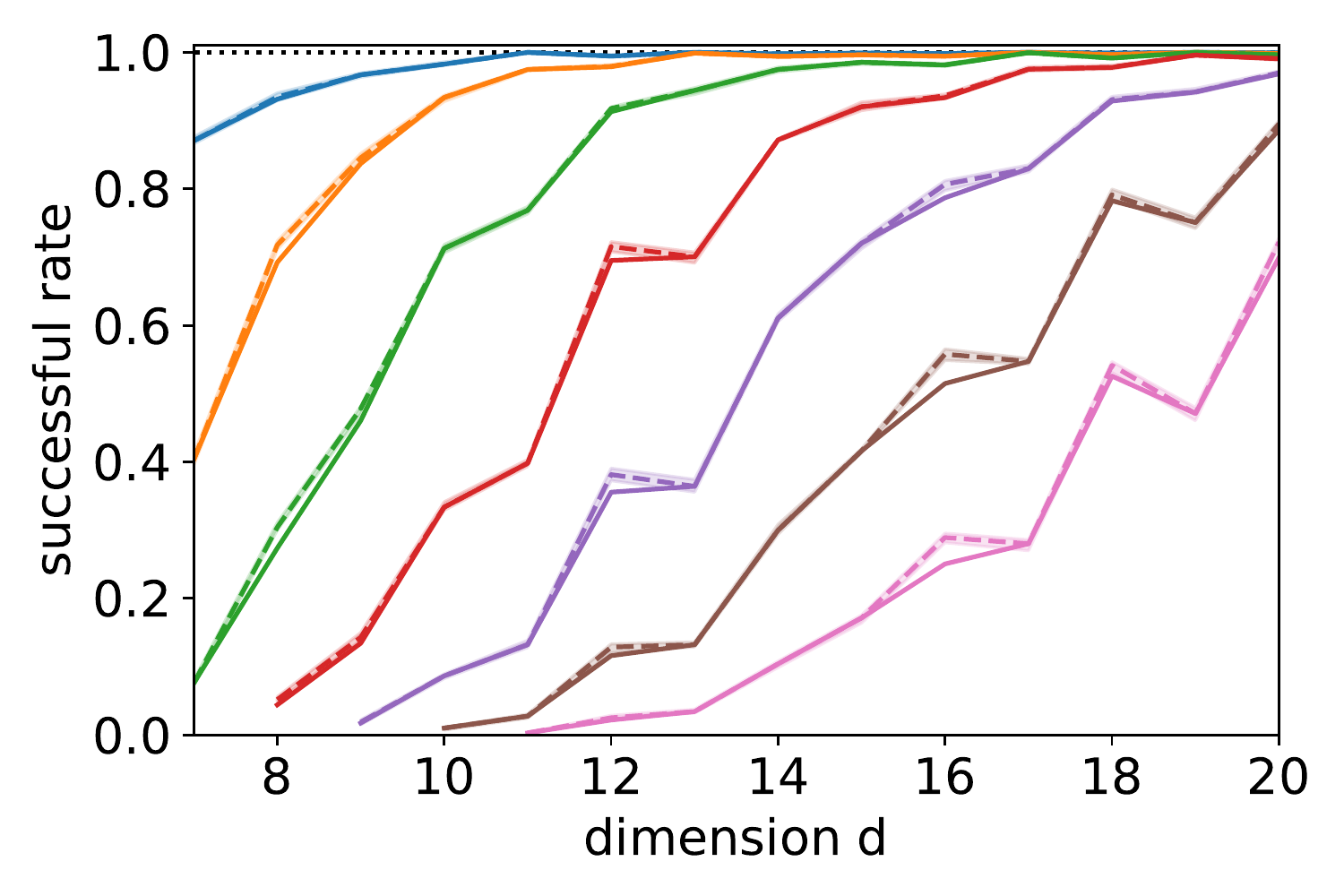}
		\caption{For $n=5,6,7,8,9,10,11$, we randomly generate $N=100000$ sets of GBSs whose cardinality are all $n$ for $d=7-20$ respectively. The solid lines represent the successful rates  by Theorem \ref{th3.1} and  the dot lines represent the successful rates  by Theorem \ref{thm:maxcom}. The lines from above to below represent  sets with cardinality  being from $5$  to $11$  respectively. Here the starting point of each curve with respect  to $n$ is with dimension $d\geq n$.}\label{fig:Sup}
	\end{figure}

	\begin{proposition}\label{prop:d=p}
		Let $p\geq 2$ be a prime. Then there are exactly $p+1$ classes of maximally commutative sets of  GBSs in $\mathbb{C}^p\otimes \mathbb{C}^p.$
	\end{proposition}
	In fact, these sets are characterized by $(0,1)\mathbb{Z}_p$, $(1,0)\mathbb{Z}_p$, and $(1,i)\mathbb{Z}_p$,  $1\leq i\leq p-1$  where $(a,b)\mathbb{Z}_p:=\{(ai,bi)| i\in\mathbb{Z}_p\}$. By this proposition and Theorem \ref{thm:maxcom}, one would   deduce Fan's result again. That is, if $\mathcal{S}$ is a set of $p$ dimensional GBSs with $l$ elements  and $l(l-1)/2\leq p$, then $\mathcal{S}$ is locally distinguishable. In fact, in this setting, the number of elements in $\Delta\mathcal{S}$ (which does not contain $(0,0)$) is less or equal than $l(l-1)/2$. However,  $(0,1)\mathbb{Z}_p\setminus\{(0,0)\}$, $(1,0)\mathbb{Z}_p\setminus\{(0,0)\}$, and $(1,i)\mathbb{Z}_p\setminus\{(0,0)\}$ ($1\leq i\leq p-1$) are $p+1$ classes of mutually disjoint sets. Therefore, there must exists some MCS $\mathcal{C}$ such that $\mathcal{C} \cap \Delta \mathcal{S}=\emptyset$. 
	
	\begin{lemma}\label{prop:sizelessthand}
		Each maximally commutative set   of GBSs in $\mathbb{C}^d\otimes\mathbb{C}^d $  must be	  with cardinality  less or equal than $d$.
	\end{lemma}
	This can be obtained by observing that commutative set of unitary matrices can be simultaneously diagonalized and the elements of GBSs are mutually orthogonal.    Moreover, one could easily verify the following lemma.
	
	\begin{lemma}\label{prop:form}
		Let $\mathcal{C}$ be a maximally commutative set of GBSs in $\mathbb{C}^d\otimes \mathbb{C}^d$. If $(i,j)$ belongs to $ \mathcal{C},$ so 
		does $(ik,jk)$ where $k\in \mathbb{Z}_d,$ i.e., $(i,j)\mathbb{Z}_d\subseteq \mathcal{C}.$    Moreover, if both $(i_1,j_1)$ and $(i_2,j_2)$ belongs to $ \mathcal{C},$   so does $(i_1+i_2,j_1+j_2).$	\end{lemma}
	From this lemma, one can conclude that each maximally commutative set  $\mathcal{C}$ can be written as the forms
	\begin{equation}\label{eq:decom_sum}
	\mathcal{C}=\bigcup_{k=1}^n (i_k,j_k)\mathbb{Z}_d,\   \text{ or } \
		\mathcal{C}=\sum_{k=1}^n (i_k,j_k)\mathbb{Z}_d.
	\end{equation}  
	Here $A+B:=\{a+b\big|\  a\in A, b\in B\}$ where $A,B$ are subsets of a group.
	
We find that the number of MCSs of GBSs is related to is an interesting function in number theory which is known as sigma function. The sigma function is usually  denoted by the Greek letter sigma ($\sigma$). This function actually denotes the sum of all divisors of a positive integer. For examples, $\sigma(6)=1+2+3+6=12,$  and  $\sigma(16)=1+2+4+8+16=31.$ Generally, let $d=p_1^{n_1}p_2^{n_2}\cdots p_l^{n_l},$ then 
	$$ \sigma(d)=\prod_{k=1}^l(1+p_k+\cdots+p_k^{n_k}). $$
	\begin{theorem}[Structure Characterization of MCS]\label{thm:structure_MCS}
		Let $d\geq 3$ be an integer. For each pair $(i,j)$ in  $\mathbb{Z}_d\times\mathbb{Z}_d$ where $i\neq 0,$ we define the following set
		$$\mathcal{C}_{i,j}:=\{(x,y)\in\mathbb{Z}_d\times\mathbb{Z}_d\big | \ \  \left|\begin{array}{cc}
			i&j\\
			x&y
		\end{array}\right|\equiv 0 \mod d,  x\in i\mathbb{Z}_d\}. 
		$$ Then $\mathcal{C}_{i,j}$ is a MCS of GBSs in $\mathbb{C}^d\otimes\mathbb{C}^d$ with exactly $d$ elements. Moreover, if we define $\mathcal{C}_{0,0}:=\{ (0,y)|y\in \mathbb{Z}_d\}$, then every MCS of GBSs in $\mathbb{C}^d\otimes\mathbb{C}^d$ must be one of $\mathcal{C}_{i,j}$  with  $i\neq 0$ or $\mathcal{C}_{0,0}.$ There are  exactly $\sigma({d})$ classes of MCSs which can be listed as follows
		$$ \mathcal{MCS}_d:=\{\mathcal{C}_{i,j}| d=ik, 0\leq j\leq k-1 \}\cup\{\mathcal{C}_{0,0}\}.$$
	\end{theorem}	
	\begin{proof}
		First, we show that the cardinality of each $\mathcal{C}_{i,j}$ is equal to $d$. Denote $d_i$ as the greatest common divisor of $i$ and $d$. Then the set $i\mathbb{Z}_d:=\{i j\in \mathbb{Z}_d \ | \ j\in \mathbb{Z}_d\}$	 has exactly $d/d_i$ elements. More exactly,
		$$i\mathbb{Z}_d=\{ ik \big| \ \ k=0,1,\cdots,\frac{d}{d_i}-1\}.$$
		For each $x=ik$ ($k=0,1,\cdots,\frac{d}{d_i}-1$), there are exactly $d_i$ solutions of $y\in\mathbb{Z}_d$ that satisfies
		\begin{equation}\label{eq:det}		       \left|\begin{array}{cc}
				i&j\\
				x&y
			\end{array}\right|\equiv 0 \mod d.
		\end{equation} 
		In fact, the Eq. \eqref{eq:det} is equivalent to 
		$i(y-kj)\equiv 0 \mod d$ whose solutions can be expressed    analytically as  $y=kj+\frac{d}{d_i}l$ where $  l=0,1,\cdots d_i-1.$  Therefore, the set $\mathcal{C}_{i,j}$ can be expressed as 
		$$
		\{(ik,kj+\frac{d}{d_i}l) \big |\  k=0,1,\cdots,\frac{d}{d_i}; l=0,1,\cdots d_i-1 \}. 
		$$ One can check that for two different pairs of $(k_1,l_1)$ and  $(k_2,l_2)$ with the above conditions, the coordinates  $(ik_1,k_1j+\frac{d}{d_i}l_1) \neq  (ik_2,k_2j+\frac{d}{d_i}l_2).$ Therefore, the cardinality of $\mathcal{C}_{i,j}$ is equal to $d.$
		
		Now we show that the elements in $\mathcal{C}_{i,j}$ are mutually commutative. In fact, for any two solutions $(ik_1,k_1j+\frac{d}{d_i}l_1) $ and $ (ik_2,k_2j+\frac{d}{d_i}l_2)$,  we  have 
		\begin{equation}\label{eq:det2}	 \left|\begin{array}{cc}
				ik_1&k_1j+\frac{d}{d_i}l_1\\[1mm]
				ik_2&k_2j+\frac{d}{d_i}l_2
			\end{array}\right|= \frac{i}{d_i}(k_1l_2-k_2l_1)d
		\end{equation}
		which is always equal to $0 \mod d$ as $d_i$ divides $i$. 
		
		Therefore, each $\mathcal{C}_{i,j}$ is a commutative set of GBSs with cardinality $d$. By Lemma \ref{prop:sizelessthand}, each  $\mathcal{C}_{i,j}$ must also be  maximally.  
		
		Next, we show that for every MCS $\mathcal{C}$,  it must be one of $\mathcal{C}_{i,j}$ with $i\neq 0$ or $C_{0,0}.$ For any maximally commutative set  $\mathcal{C}$ of GBSs in $\mathbb{C}^d\otimes \mathbb{C}^d$, by Eq. \eqref{eq:decom_sum}, there exists $(x_k,y_k)\in \mathcal{C}$ ($k=1,\cdots,n$), such that 
		$$	 \mathcal{C}=\sum_{k=1}^n (x_k,y_k)\mathbb{Z}_d.$$
		If all $x_k$ are equal to 0, one must conclude that $\mathcal{C}=\mathcal{C}_{0,0}.$ If not, let $i$ denote the greatest common divisor of $x_1,x_2,\cdots,x_n$ and $d$, which is not equal to zero in this case. There exist $r_k\in\mathbb{Z}_d$, such that $
		i=\sum_{k=1}^nr_kx_k$ (by Ref. \cite{Nathanson2000}, p12, Theorem 1.4, we have  $i=R_0d+\sum_{k=1}^nR_kx_k,  R_i\in \mathbb{Z}$, then taking module $d$). And we define $j=\sum_{k=1}^nr_ky_k.$ By Lemma \ref{prop:form}, we have $(i,j)\in \mathcal{C}.$ As both $(x_k,y_k)$ and $(i,j) $ are in $\mathcal{C}.$  By the definition of $i$, for each $k$, the element $x_k\in i\mathbb{Z}_d$. As both $(x_k,y_k)$ and $(i,j) $ are in $\mathcal{C}$, we have $iy_k-jx_k\equiv 0 \mod d.$ Therefore, by definition of $\mathcal{C}_{i,j}$, for each $k$, the element $(x_k,y_k)\in \mathcal{C}_{i,j}.$ By Lemma \ref{prop:form} again, one have $\mathcal{C}\subseteq \mathcal{C}_{i,j}.$ However, both sets are maximally commutative sets of GBSs. Therefore, $\mathcal{C}$ must equal to $\mathcal{C}_{i,j}$.
		
		In the following, we show that each $\mathcal{C}_{x,y}$ ($x\neq0$) is in fact lie in one of $\mathcal{MCS}_d$. Set $d_x$ denote the greatest common divisor of $x$ and $d$ (we might assume $x=c_xd_x $ where $c_x\in \mathbb{Z}$). So $d_x=qx+rd$ for some integers $q,r$.   There exist a unique  $j\in\{0,1,\cdots,k_x-1\}$ (where $k_xd_x=d$) such that
		$$qy-j\in k_x\mathbb{Z}_d.$$
		That is, $qy-j=k_xl_x$ for some $l_x\in\mathbb{Z}_d.$ For this $j$, we have  the following equation
		$$
		\left|\begin{array}{cc}
			d_x&j\\
			x&y
		\end{array}\right|=
		\left|\begin{array}{cc}
			d_x&qy-k_xl_x\\
	x	&y
		\end{array}\right|=(d_x-qx)y+l_xk_x x
		$$
		which is equal $(ry+l_xc_x)d \equiv 0$   under   $\mod  d.$ By definition, $x\in d_x\mathbb{Z}_d$. Therefore, we have $(x,y)\in \mathcal{C}_{d_x,j}.$   As the elements in  $\mathcal{C}_{d_x,j}$ are commute with each other, for any $(x_1,y_1)\in \mathcal{C}_{d_x,j}$,  we have 
		$$
		\left|\begin{array}{cc}
			x&y\\
			x_1&y_1
		\end{array}\right|\equiv 0 \mod d.$$ Note that $d_x\mathbb{Z}_d=x\mathbb{Z}_d$, we have $x_1\in x\mathbb{Z}_d.$  Therefore, $\mathcal{C}_{d_x,j}\subseteq \mathcal{C}_{x,y}.$ By the maximality, we must have $ \mathcal{C}_{x,y}=\mathcal{C}_{d_x,j}.$ Therefore, one conclude that every $\mathcal{C}_{x,y}$ must be one of the elements in $\mathcal{MCS}_d.$
		
		On the other hand, we need to show that the sets in $\mathcal{MCS}_d$ are mutually different. Clearly, $\mathcal{C}_{0,0}$ is different from all the other sets. Let $\mathcal{C}_{i_1,j_1}$ and  $\mathcal{C}_{i_2,j_2}$ be any two members of  $\mathcal{MCS}_d$ where $(i_1,j_1)\neq (i_2,j_2)$ and $i_1,i_2$ are nonzero. As $d=i_1k_1,$ if $i_1=i_2,$  we have $i_1j_2-i_2j_1=i_1 (j_2- j_1)$ which lies   between $-(d-1)$ and $d-1$  but not equal to zero. Hence, $(i_1,j_1)$ and $(i_2,j_2)$ are not commute. Therefore,  $\mathcal{C}_{i_1,j_1}\neq \mathcal{C}_{i_2,j_2}.$ If $i_1\neq i_2$, we can assume that $i_1<i_2$ without loss of generality. As both $i_1$ and $i_2$ are divisors of $D$, one can check that $i_1\notin i_2\mathbb{Z}_d$.	Therefore, by definition, $(i_1,j_1)\notin\mathcal{C}_{i_2,j_2}.$ Hence, we also have 	   $\mathcal{C}_{i_1,j_1}\neq \mathcal{C}_{i_2,j_2}.$
		
		For each divisor $i (1\leq i<d)$ of $d$, it contributes to $d/i$ classes of MCSs to $\mathcal{MCS}_d.$  Therefore,
		$$|\mathcal{MCS}_d|=1+\sum_{i|d, 1\leq i<d} d/i=\sum_{i|d} d/i=\sigma(d).$$
		This completes the proof.
	\end{proof}

	\begin{figure}[h]
		\centering
		\includegraphics[scale=0.59]{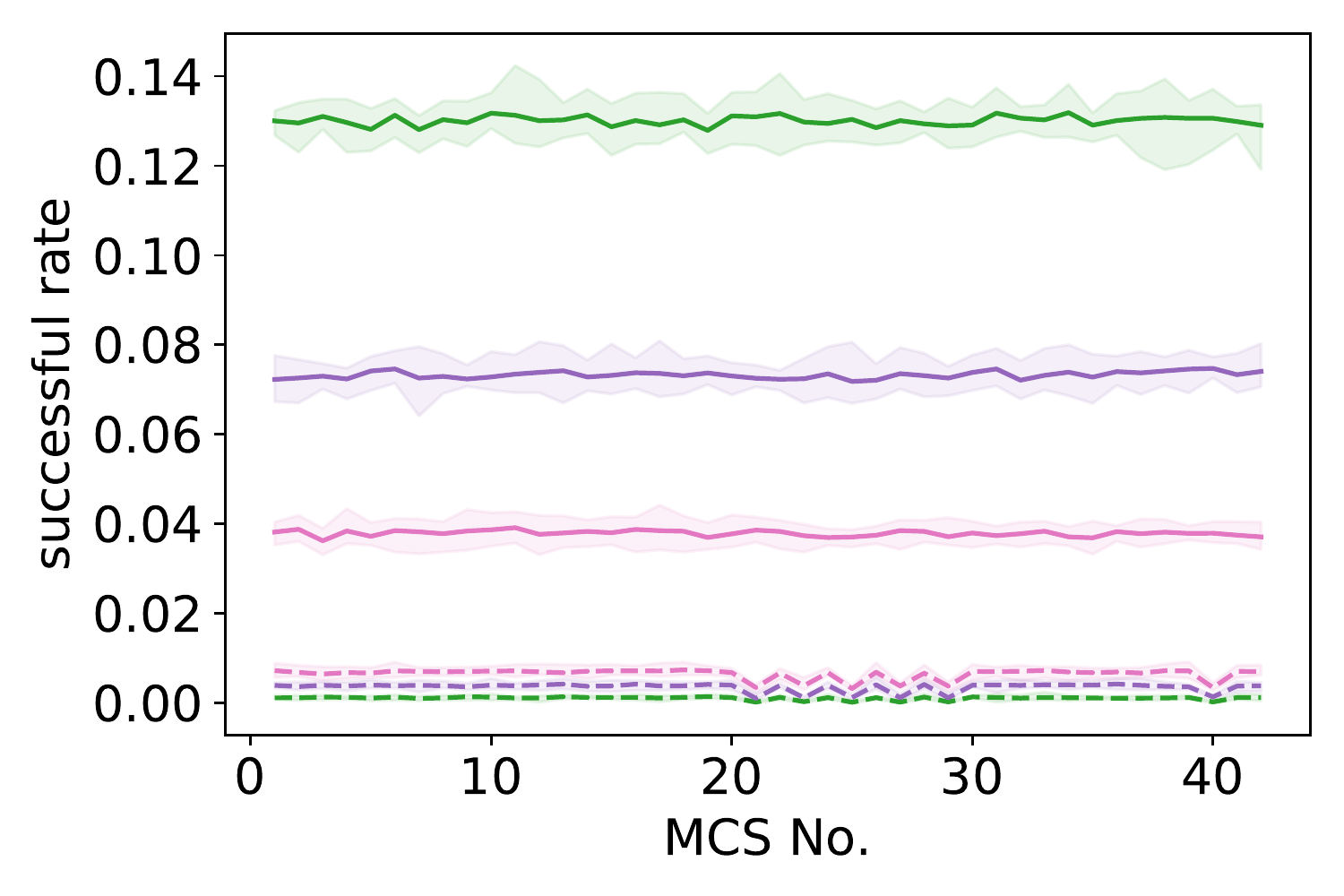}
		\caption{We consider the cases $d=20$ and $n=9,10$ or $11$ (which correspond to curves with color green, purple and pink respectively).  We randomly generate 100000 sets (which are separated into 10 equal classes) of GBSs whose cardinality are all $n=9,10,11$  respectively.  Each solid curve shows the successful rates for the $\sigma(20)=42$ classes of MCSs. The dot lines represent the successful rates for the $42$ classes of MCSs such that  the local distinguishability of the samples can only detected by one class of MCS itself.    }\label{fig:MCSSuccessp}
	\end{figure}

From the above theorem, we know that there are $\sigma(d)$ classes of MCSs  of $d$ dimensional 
GBSs.  	 Are there any differences in the ability of these MCSs to detect generalized Bell state sets?   
Is there any redundant MCS in detecting the local discrimination of generalized Bell sets? We present some numerical results for the two questions.

	The three solid lines in Fig. \ref{fig:MCSSuccessp} imply that the successful rates of all MCSs  are almost equal to each other (one should compare this with the detectors defined in Eq. \eqref{PointDe}, see Fig. \ref{fig:PointSucess}). The three dot lines imply that each class of MCSs is irredudant in the sense that  for each MCS $\mathcal{C}$, there exist some set $\mathcal{S}$ whose local distinguishability can only be detected by  $\mathcal{C}$ but not by other MCSs.
	
		\begin{figure}[h]
		\centering
		\includegraphics[scale=0.59]{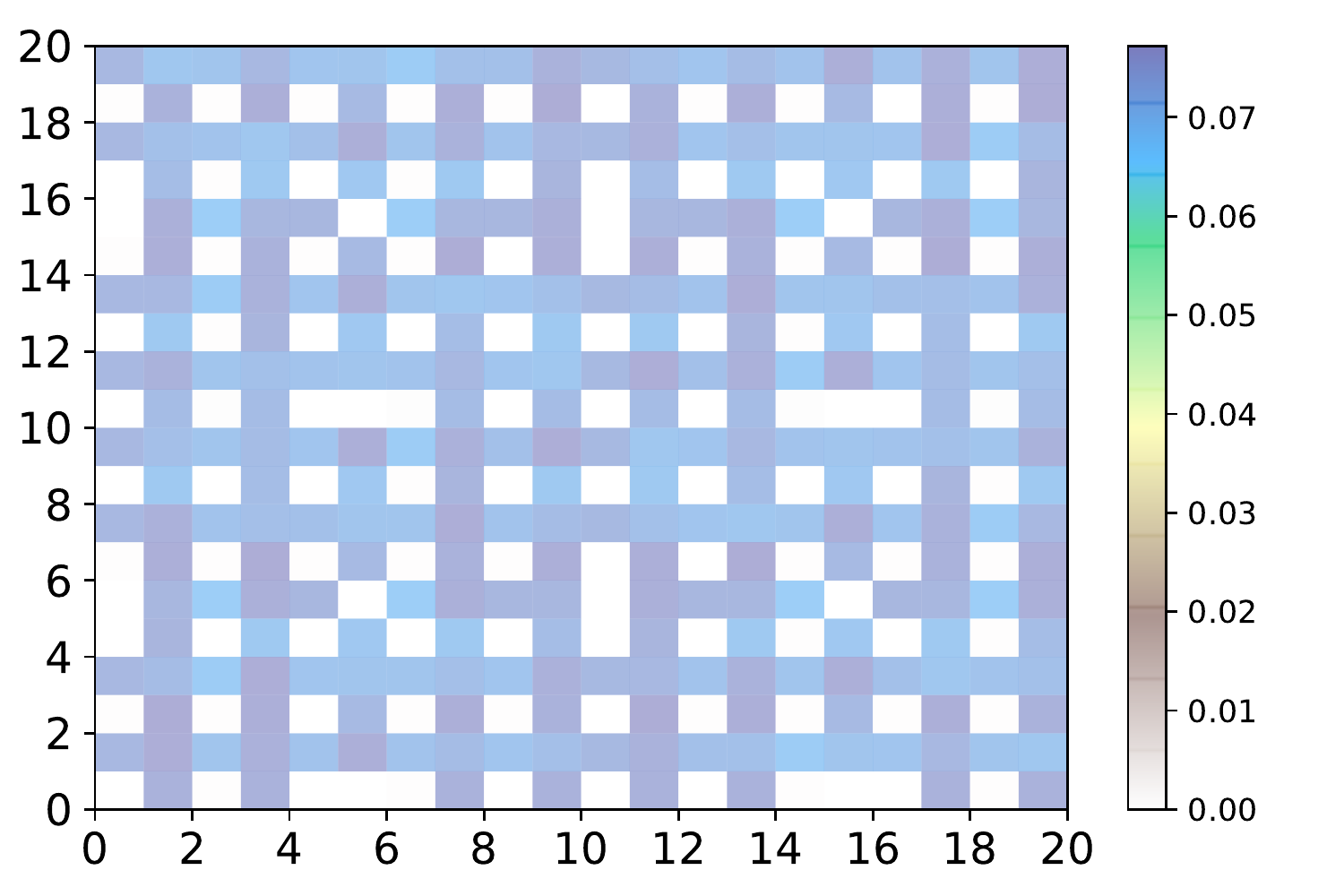}
		\caption{We consider the case $d=20$ and $n=10$. The figure shows the successful rates for each detectors (see Eq. \eqref{PointDe}) indicated by the coordinates $(i,j)\in\mathbb{Z}_{20}\times \mathbb{Z}_{20}$. We randomly generate 100000 sets   of GBSs whose cardinality are all $n=10$.    }\label{fig:PointSucess}
	\end{figure}
	\section{Conclusion and discussion}\label{Sec:Con_Dis}
	
	In this paper, we studied the problem of local distinguishability of   generalized Bell states.  Firstly, we gave a review of some important methods for detecting the local distinguishability of   GBSs. Motivated by a recent method derived by Yuan et. al. \cite{Yuan21}, we introduced the concept maximally commutative set of GBSs. Surprisingly, we found that each MCS is useful for detecting the local distinguishability of GBSs. More exactly, given a set $\mathcal{S}$ of  GBSs, if there exists some MCS $\mathcal{C}$ such that  the difference set  of  $\mathcal{S}$ is disjoint with or contains in $\mathcal{C},$ then the set $\mathcal{S}$ can be one-way distinguishable.   This method is stronger than that in Ref. \cite{Yuan21}.  This motivates us to find out all the MCSs of given dimension.  Indeed, we presented a complete structure characterization of MCS in Theorem \ref{thm:structure_MCS}. 
	
	However, MCS only gives a sufficient condition for locally distinguishable, it is not necessary. It is interesting to derive an easy checking condition for local distinguishability of GBSs which are both sufficient and necessary. In addition, it is interesting to check whether Fan's results can be extended to systems without the assumption on the dimension of local systems. A weaker form is that: given any integer $l$, do there exist some $D$ (which depends on $l$) such that if $d\geq D$, then any $l$ GBSs in $\mathbb{C}^d\otimes \mathbb{C}^d$ are locally distinguishable? As far as we known, this problem is only solved for the case $l=3$. We conjecture that this holds for all other cases. 
	
	Note that our method here can be generalized to any maximally entangled basis whose defining unitary matices $\mathcal{B}$   satisfies: for any $U,V\in \mathcal{B}$ there exists some $W\in \mathcal{B}$ such that $U^\dagger V\propto W$.   The lattice qudits basis \cite{Tian15} is such an example.  From their proof, any locally distinguishable  set of lattices qudits basis  that can be detected by Ref. \cite{Tian15}   can be always detected by a MCS of lattice qudits basis.  Therefore, our method can be also seen as a generalization of theirs. Therefore, it is also interesting to give a complete characterization of the MCS of lattice qudits basis and study its application to local discrimination.

	\begin{acknowledgments}
	M.S.L. and Y.L.W. were supported  by  National  Natural  Science  Foundation  of  China under Grant No. (12005092, 11901084, 61773119),  the China Postdoctoral Science Foundation (2020M681996),   the Research startup funds of DGUT (GC300501-103). F.S. was supported by the NSFC under Grants No. 11771419 and No.12171452, the Anhui Initiative in Quantum Information Technologies under Grant No. AHY150200, and the National Key Research and Development Program of China 2020YFA0713100.
	\end{acknowledgments}


\begin{thebibliography}{}
		
		
		\bibitem{nils04} M. A. Nielsen and I. L. Chuang.  \emph{Quantum Computation and Quantum Information}. Cambridge University Press, 2004.
		
		
		
		\bibitem{Bennett99}
	C. H. Bennett, D. P. DiVincenzo, C. A. Fuchs, T. Mor, E. Rains, P. W. Shor, J. A. Smolin, and W. K. Wootters, Quantum nonlocality without entanglement, \href{https://doi.org/10.1103/PhysRevA.59.1070}{Phys. Rev. A {\bf 59}, 1070(1999)}.
	
	\bibitem{Walgate-2000} J. Walgate, A. J. Short, L. Hardy, and V.
			Vedral, Local distinguishability of multipartite orthogonal quantum
			states, \href{https://doi.org/10.1103/PhysRevLett.85.4972}{Phys. Rev. Lett. {\bf85}, 4972 (2000)}.
				
	 
			
	
		\bibitem{Terhal01} B. M. Terhal, D. P. DiVincenzo, and D. W. Leung, Hiding Bits
		in Bell States, \href{https://doi.org/10.1103/PhysRevLett.86.5807}{Phys. Rev. Lett. \textbf{86}, 5807 (2001).}
		
		
 
		
		\bibitem{DiVincenzo02}
		D. P. DiVincenzo, D.W.  Leung and B.M. Terhal, Quantum data hiding,
		\href{https://ieeexplore.ieee.org/document/985948/}{IEEE Trans. Inf. Theory  \textbf{48}, 580 (2002)}.
		
		\bibitem{Markham08}
		D. Markham  and B. C. Sanders, Graph States for Quantum Secret Sharing,
		\href{https://doi.org/10.1103/PhysRevA.78.042309}{Phys. Rev. A  \textbf{78}, 042309 (2008)}.
		
		
		
		\bibitem{Rahaman15}  R. Rahaman and M. G. Parker, Quantum scheme for secret
		sharing based on local distinguishability, \href{https://doi.org/10.1103/PhysRevA.91.022330}{Phys. Rev. A \textbf{91},
			022330 (2015).}
		
		\bibitem{WangJ17} J. Wang, L. Li, H. Peng, and Y. Yang, Quantum-secret-sharing
		scheme based on local distinguishability of orthogonal multiqudit entangled states, \href{https://doi.org/10.1103/PhysRevA.95.022320}{Phys. Rev. A \textbf{95}, 022320 (2017).}
		
 
			
			\bibitem{ben99u} C. H. Bennett, D. P. DiVincenzo, T. Mor, P. W. Shor,
			J. A. Smolin, and B. M. Terhal, Unextendible Product Bases and Bound
			Entanglement, \href{https://doi.org/10.1103/PhysRevLett.82.5385}{ Phys. Rev. Lett. {\bf82}, 5385 (1999)}.
			
			
			
			
			\bibitem{Ghosh-2001} S. Ghosh, G. Kar, A. Roy, A. Sen (De), and U.
			Sen, Distinguishability of Bell states, \href{https://doi.org/10.1103/PhysRevLett.87.277902}{Phys. Rev. Lett. {\bf87}, 277902 (2001)}.
			
			 
			\bibitem{walgate-2002} J. Walgate and L. Hardy, Nonlocality, asymmetry,
			and distinguishing bipartite states, \href{https://doi.org/10.1103/PhysRevLett.89.147901}{Phys. Rev. Lett. {\bf89}, 147901 (2002)}.
			
			\bibitem{HSSH} M. Horodecki, A. Sen(De), U. Sen, and K. Horodecki,
			Local indistinguishability: more nonlocality with less entanglement,
			\href{https://doi.org/10.1103/PhysRevLett.90.047902}{Phys. Rev. Lett. {\bf90}, 047902 (2003)}.
			
			\bibitem{divin03} D. P. DiVincenzo, T. Mor, P. W. Shor, J. A. Smolin,
			and B. M. Terhal, Unextendible product bases, uncompletable product
			bases and bound entanglement, \href{https://doi.org/10.1007/s00220-003-0877-6}{Commun. Math. Phys. {\bf238}, 379 (2003)}.
			
			\bibitem{rin04} S. De Rinaldis, Distinguishability of complete and
			unextendible product bases, \href{https://doi.org/10.1103/PhysRevA.70.022309}{Phys. Rev. A {\bf70}, 022309 (2004)}.
			
			\bibitem{Ghosh-2004} S. Ghosh, G. Kar, A. Roy, and D. Sarkar, Distinguishability
			of maximally entangled states, \href{https://doi.org/10.1103/PhysRevA.70.022304}{Phys. Rev. A {\bf70}, 022304 (2004)}. 
			
			\bibitem{fan-2005} H. Fan, Distinguishability and indistinguishability
			by local operations and classical communication, \href{https://doi.org/10.1103/PhysRevLett.92.177905}{Phys. Rev. Lett. {\bf92}, 177905 (2004)}.
			
			\bibitem{Nathanson-2005} M. Nathanson, Distinguishing bipartite orthogonal
			states by LOCC: best and worst cases, \href{https://doi.org/10.1063/1.1914731}{Journal of Mathematical Physics {\bf46}, 062103 (2005)}.
			
			\bibitem{Watrous-2005} J. Watrous, Bipartite subspaces having no
			bases distinguishable by local operations and classical communication,
			\href{https://doi.org/10.1103/PhysRevLett.95.080505}{Phys. Rev. Lett. {\bf95}, 080505 (2005)}.
			
			\bibitem{Hayashi-etal-2006} M. Hayashi, D. Markham, M. Murao, M.
			Owari, and S. Virmani, Bounds on entangled orthogonal state discrimination
			using local operations and classical communication, \href{https://doi.org/10.1103/PhysRevLett.96.040501}{Phys. Rev. Lett. {\bf96}, 040501 (2006)}.
			
			 
			
			\bibitem{nis06} J. Niset and N. J. Cerf, Multipartite nonlocality
			without entanglement in many dimensions, \href{https://doi.org/10.1103/PhysRevA.74.052103}{Phys. Rev. A {\bf74}, 052103 (2006)}.
			
			\bibitem{Duan2007} R. Y. Duan, Y. Feng, Z. F. Ji, and M. S. Ying,
			Distinguishing arbitrary multipartite basis unambiguously using local
			operations and classical communication, \href{https://doi.org/10.1103/PhysRevLett.98.230502}{Phys. Rev. Lett. {\bf98}, 230502 (2007)}.
			
			\bibitem{feng09} Y. Feng and Y.-Y. Shi, Characterizing locally indistinguishable
			orthogonal product states, \href{https://doi.org/10.1109/TIT.2009.2018330}{IEEE Trans. Inf. Theory {\bf55}, 2799 (2009)}.
			
			\bibitem{Duan-2009} R. Y. Duan, Y. Feng, Y. Xin, and M. S. Ying,
			Distinguishability of quantum states by separable operations, \href{https://doi.org/10.1109/TIT.2008.2011524}{IEEE Trans. Inform. Theory {\bf55}, 1320 (2009)}. 
			
			
			
			\bibitem{BGK-2011} S. Bandyopadhyay, S. Ghosh and G. Kar, LOCC distinguishability
			of unilaterally transformable quantum states, \href{https://doi.org/10.1088/1367-2630/13/12/123013}{New J. Phys. {\bf13}, 123013 (2011)}.
			
			\bibitem{Bandyo-2011} S. Bandyopadhyay, More nonlocality with less
			purity, \href{https://doi.org/10.1103/PhysRevLett.106.210402}{Phys. Rev. Lett. {\bf106}, 210402 (2011)}.
			
			\bibitem{Yu-Duan-2012} N. Yu, R. Duan, and M. Ying, Four Locally
			Indistinguishable Ququad-Ququad Orthogonal Maximally Entangled States,
			\href{https://doi.org/10.1103/PhysRevLett.109.020506}{Phys. Rev. Lett. {\bf109}, 020506 (2012)}.
			
			\bibitem{Cosentino-2013} A. Cosentino, Positive partial transpose indistinguishable
			states via semidefinite programming, \href{https://doi.org/10.1103/PhysRevA.87.012321}{Phys. Rev. A {\bf87}, 012321 (2013)}.
			
			
			\bibitem{Nathanson13}M. Nathanson, Three maximally entangled states can require two-way local operations and
classical communication for local discrimination, \href{https://journals.aps.org/pra/abstract/10.1103/PhysRevA.88.062316}{Phys. Rev. A \textbf{88}, 062316 (2013).}
 
			
			\bibitem{Li15} M.-S. Li, Y.-L. Wang, S.-M. Fei and Z.-J. Zheng, $d$ locally indistinguishable maximally entangled states in $\mathbb{C}^d\otimes\mathbb{C}^d$, \href{https://doi.org/10.1103/PhysRevA.91.042318}{Phys.
				Rev. A \textbf{91}, 042318 (2015)}.		
			
			
			\bibitem{Yus15}  S. X. Yu and C. H. Oh, Detecting the local indistinguishability of maximally entangled states, \href{http://arxiv.org/abs/arXiv:1502.01274v1}{arXiv:1502.01274v1}.
			
			
			
			\bibitem{Li20}M.-S. Li, S.-M. Fei, Z.-X. Xiong,  and Y.-L. Wang, Twist-teleportation-based local discrimination of maximally entangled states,  \href{https://link.springer.com/article/10.1007\%2Fs11433-020-1562-4}{SCIENCE CHINA Physics, Mechnics $\&$  Astronomy \textbf{63} 8,  280312 (2020).}
	
			
			\bibitem{BN-2013} S. Bandyopadhyay, M. Nathanson, Tight bounds on
			the distinguishability of quantum states under separable measurements,
			\href{https://doi.org/10.1103/PhysRevA.88.052313}{Phys. Rev. A {\bf88}, 052313 (2013)}.
			
			\bibitem{Cosentino-Russo-2014} A. Cosentino and V. Russo, Small sets
			of locally indistinguishable orthogonal maximally entangled states,
			\href{https://doi.org/10.26421/QIC14.13-14}{Quantum Information and Computation {\bf14}, 1098 (2014)}.
			
			\bibitem{B-IQC-2015} S. Bandyopadhyay, A. Cosentino, N. Johnston,
			V. Russo, J. Watrous, and N. Yu, Limitations on separable measurements
			by convex optimization, \href{https://doi.org/10.1109/TIT.2015.2417755}{IEEE Transactions on Information Theory, {\bf61},  3593 (2015)}.. 
			
			
			
			
			\bibitem{childs13} A. M. Childs, D. Leung, L. Man\v{c}inska, and
			M. Ozols, A framework for bounding nonlocality of state discrimination,
			\href{https://doi.org/10.1007/s00220-013-1784-0}{Commun. Math. Phys. {\bf323}, 1121 (2013)}. 
			
			\bibitem{Yang13} Y.-H. Yang, F. Gao, G.-J. Tian, T.-Q. Cao, and Q.-Y.
			Wen, Local distinguishability of orthogonal quantum states in a $2\otimes2\otimes2$
			system, \href{https://doi.org/10.1103/PhysRevA.88.024301}{Phys. Rev. A {\bf88}, 024301 (2013)}. 
			
			\bibitem{zhang14} Z.-C. Zhang, F. Gao, G.-J. Tian, T.-Q. Cao, and
			Q.-Y. Wen, Nonlocality of orthogonal product basis quantum states,
			\href{https://doi.org/10.1103/PhysRevA.90.022313}{Phys. Rev. A {\bf90}, 022313 (2014)}. 
			
			\bibitem{zhang15} Z.-C. Zhang, F. Gao, S.-J. Qin, Y.-H. Yang, and
			Q.-Y. Wen, Nonlocality of orthogonal product states, \href{https://doi.org/10.1103/PhysRevA.92.012332}{Phys. Rev. A {\bf92}, 012332 (2015)}.
			
			\bibitem{wang15} Y.-L. Wang, M.-S. Li, Z.-J. Zheng, and S.-M. Fei,
			Nonlocality of orthogonal product-basis quantum states, \href{https://doi.org/10.1103/PhysRevA.92.032313}{Phys. Rev. A {\bf92}, 032313 (2015)}.
			
			 
			
		 
			
			\bibitem{Xu-16-2} G.-B. Xu, Q.-Y. Wen, S.-J. Qin, Y.-H. Yang, and
			F. Gao, Quantum nonlocality of multipartite orthogonal product states,
			\href{https://doi.org/10.1103/PhysRevA.93.032341}{Phys. Rev. A {\bf93}, 032341 (2016)}.
			
			\bibitem{zhang16} Z.-C. Zhang, F. Gao, Y. Cao, S.-J. Qin, and Q.-Y.
			Wen, Local indistinguishability of orthogonal product states, \href{https://doi.org/10.1103/PhysRevA.93.012314}{Phys. Rev. A {\bf93}, 012314 (2016)}.
			
			\bibitem{zhang16-1} X. Zhang, X. Tan, J. Weng, and Y. Li, LOCC indistinguishable
			orthogonal product quantum states, \href{https://doi.org/10.1038/srep28864}{Sci. Rep. {\bf6}, 28864 (2016)}.
			
			
			\bibitem{Wang-2017-Qinfoprocess} Y.-L. Wang, M.-S. Li, Z.-J. Zheng,
			and S.-M. Fei, The local indistinguishability of multipartite product
			states, \href{https://doi.org/10.1007/s11128-016-1477-7}{Quantum Inf. Processing {\bf16}, 5 (2017)}.
			
			\bibitem{Zhang-Oh-2017} Z.-C. Zhang, K.-J. Zhang, F. Gao, Q.-Y. Wen,
			and C. H. Oh, Construction of nonlocal multipartite quantum states,
			\href{https://doi.org/10.1103/PhysRevA.95.052344}{Phys. Rev. A {\bf95}, 052344 (2017)}. 
			
				\bibitem{halder} S. Halder, Several nonlocal sets of multipartite
			pure orthogonal product states, \href{https://doi.org/10.1103/PhysRevA.98.022303}{Phys. Rev. A {\bf98}, 022303 (2018)}.
			
			\bibitem{Li18}M.-S. Li and Y.-L. Wang, Alternative method for deriving nonlocal multipartite product states, \href{https://journals.aps.org/pra/abstract/10.1103/PhysRevA.98.052352}{Phys. Rev. A \textbf{98}, 052352 (2018).}
			
			\bibitem{Xu20b} D.-H. Jiang, and G.-B. Xu, Nonlocal sets of orthogonal product states in arbitrary multipartite quantum system, \href{https://journals.aps.org/pra/abstract/10.1103/PhysRevA.102.032211}{Phys. Rev. A  \textbf{102}, 032211 (2020).}
			
			
				
			
			\bibitem{Halder1909}S. Halder, and C. Srivastava, Locally distinguishing quantum states with limited classical communication,  	\href{10.1103/PhysRevA.101.052313}{Phys. Rev. A {\bf 101}, 052313 (2020).}	
			
			\bibitem{Halder20c}	S. Halder, R. Sengupta, 	Distinguishability classes, resource sharing, and bound entanglement distribution,  \href{https://journals.aps.org/pra/abstract/10.1103/PhysRevA.101.012311}{Phys. Rev. A \textbf{101},  012311 (2020).}
			
			\bibitem{Halder19}S. Halder, M. Banik, S. Agrawal, and S. Bandyopadhyay, Strong Quantum Nonlocality without Entanglement, 	\href{https://journals.aps.org/prl/abstract/10.1103/PhysRevLett.122.040403}{Phys. Rev. Lett. \textbf{122}, 040403 (2019).}
			
			\bibitem{Zhang1906}Z.-C. Zhang and X. Zhang, Strong quantum nonlocality in multipartite quantum systems, \href{https://journals.aps.org/pra/abstract/10.1103/PhysRevA.99.062108}{Phys. Rev. A \textbf{99}, 062108 (2019).}
			
			\bibitem{Shi20S}     F. Shi, M. Hu, L. Chen, and X. Zhang, Strong quantum nonlocality with entanglement, \href{https://journals.aps.org/pra/abstract/10.1103/PhysRevA.102.042202}{Phys. Rev. A \textbf{102},  042202 (2020).}
			
			\bibitem{Tian20} P. Yuan, G. J. Tian, and X. M. Sun, Strong quantum nonlocality without entanglement in multipartite quantum systems,	\href{https://journals.aps.org/pra/abstract/10.1103/PhysRevA.102.042228}{Phys. Rev. A \textbf{102},  042228 (2020).}
			
			\bibitem{Wang21}	  Y.-L.   Wang,   M.-S.   Li,   and   M.-H.   Yung,    Graph-connectivity-based strong quantum nonlocality with genuine entanglement, \href{https://journals.aps.org/pra/abstract/10.1103/PhysRevA.104.012424}{Phys. Rev. A 104, 012424 (2021). } 
			
			
			
			
			\bibitem{Banik21} M. Banik, T. Guha, M. Alimuddin, G. Kar, S. Halder, and S. S. Bhattacharya, Multicopy Adaptive Local Discrimination: Strongest Possible Two-Qubit Nonlocal Bases, \href{https://journals.aps.org/prl/abstract/10.1103/PhysRevLett.126.210505}{Phys. Rev. Lett. \textbf{126}, 210505 (2021).}
			
		 
			
			 
		
		
			\bibitem{Ha21}	D. Ha and Y. Kwon, Quantum nonlocality without entanglement: explicit dependence on prior probabilities of nonorthogonal mirror-symmetric states,  \href{https://www.nature.com/articles/s41534-021-00415-0}{npj Quantum Inf. \textbf{7} 81 (2021).}		
		
		
	
		
	
		
		
		
	
		
 
		\bibitem{Tian15} G.-J. Tian,   S.-X. Yu,  F. Gao,  Q.-Y. Wen and  C.H. Oh,
		Local discrimination of qudit lattice states via commutativity,
		\href{https://journals.aps.org/pra/abstract/10.1103/PhysRevA.92.042320 }{Phys. Rev. A  \textbf{92}, 042320 (2015).}	
		
	 
 
 
		\bibitem{Tian15_2} G.-J. Tian,   S.-X. Yu,  F. Gao,  Q.-Y. Wen and  C.H. Oh,		
		Local discrimination of four or more maximally entangled states,
		\href{https://journals.aps.org/pra/abstract/10.1103/PhysRevA.91.052314 }{Phys. Rev. A \textbf{91}, 052314 (2015).}
		
	
		
		

		
	
		
		\bibitem{Tian16} G.-J. Tian,   S.-X. Yu,  F. Gao and  Q.-Y. Wen,
		Classification of locally distinguishable and indistinguishable sets of maximally entangled states,
		\href{https://journals.aps.org/pra/abstract/10.1103/PhysRevA.94.052315}{Phys. Rev. A \textbf{94}, 052315 (2016).}
		
		\bibitem{Singal15}
		T. Singal, R. Rahaman, S. Ghosh, G. Kar,
		Necessary condition for local distinguishability of maximally entangled states: Beyond orthogonality preservation,	\href{https://journals.aps.org/pra/abstract/10.1103/PhysRevA.96.042314}{Phys. Rev. A  \textbf{96}, 042314 (2017).}
		
		\bibitem{Wang17}	
		Y.-L. Wang,  M.-S. Li,  S.-M. Fei and  Z.-J. Zheng,  The local distinguishability of any three generalized Bell states, \href{https://link.springer.com/article/10.1007/s11128-017-1579-x}{Quantum Inf. Process.  \textbf{16}, 126 (2017).}
		
		
	


	\bibitem{Yuan20}  J. T. Yuan, Y. H. Yang and C. H. Wang, Constructions of
locally distinguishable sets of maximally entangled states
which require two-way LOCC, \href{https://iopscience.iop.org/article/10.1088/1751-8121/abc43b}{J. Phys. A: Math. Theor.
\textbf{53}, 505304 (2020).}


	\bibitem{Yang21}Y. H. Yang, G. F. Mu, J. T. Yuan and C. H. Wang,
Distinguishability of generalized Bell states in arbitrary
dimension system via one-way local operations and classical communication, \href{https://link.springer.com/article/10.1007/s11128-021-02990-9}{Quant. Info. Proc. \textbf{20}, 52 (2021).}


\bibitem{Hashimoto21}	T. Hashimoto, M. Horibe, and A. Hayashi, 	Simple criterion for local distinguishability of generalized Bell states in prime dimension, \href{https://journals.aps.org/pra/abstract/10.1103/PhysRevA.103.052429}{Phys. Rev. A \textbf{103}, 052429 (2021).}
		
		
	\bibitem{Yuan21} J.-T. Yuan, Y.-H. Yang, C.-H. Wang, Necessary and sufficient conditions for local discrimination of generalized Bell states: finding out all locally indistinguishable sets of generalized Bell states, \href{https://arxiv.org/abs/2109.07390}{arXiv:2109.07390.}
	
	
		\bibitem{Nathanson2000}	Melvyn B. Nathanson, \emph{Elementary Methods in Number Theory},  Graduate Texts in Mathematics 195, 1 ed, 	Springer Verlag New York 2000.
	\end{thebibliography}
\end{document}